\documentclass{amsart}
\usepackage{amsbsy,amssymb,amscd,amsfonts,latexsym,amstext,delarray, amsmath,graphicx,color,caption,enumerate,amsaddr}
\usepackage{graphicx}
\usepackage{epsfig}
\input xy

\xyoption{all}
\usepackage{euscript}
\usepackage{multirow}
\usepackage{etex, pictexwd,dcpic}
\usepackage[makeroom]{cancel}

\newtheorem{theorem}{Theorem}
\newtheorem{definition}[theorem]{Definition}
\newtheorem{proposition}[theorem]{Proposition}
\newtheorem{lemma}[theorem]{Lemma}
\newtheorem{corollary}[theorem]{Corollary}

\newcommand{\db}[1]{\mathopen{\{\!\!\{}#1\mathopen{\}\!\!\}}}

\begin{document}
\author{Semeon Arthamonov}
\address{Department of Mathematics, Rutgers, The State University of New Jersey,\\
110 Frelinghuysen Rd., Piscataway, NJ 08854, USA.
}
\email{semeon.artamonov@rutgers.edu}
\title{Noncommutative Inverse Scattering Method for the Kontsevich system}
\begin{abstract}
We formulate an analog of Inverse Scattering Method for integrable systems on noncommutative associative algebras. In particular we define Hamilton flows, Casimir elements and noncommutative analog of the Lax matrix. The noncommutative Lax element generates infinite family of commuting Hamilton flows on an associative algebra. The proposed approach to integrable systems on associative algebras satisfy certain universal property, in particular it incorporates both classical and quantum integrable systems as well as provides a basis for further generalization.

We motivate our definition by explicit construction of noncommutative analog of Lax matrix for a system of differential equations on associative algebra recently proposed by Kontsevich. First we present these equations in the Hamilton form by defining a bracket of Loday type on the group algebra of the free group with two generators. To make the definition more constructive we utilize (with certain generalizations) the Van den Bergh approach to Loday brackets via double Poisson brackets. We show that there exists an infinite family of commuting flows generated by the noncommutative Lax element.
\end{abstract}

\maketitle


\tableofcontents

\section{Introduction}

An idea of studying Integrable Systems on associative algebras was originated by Gelfand school in the late 1970s \cite{GelfandDorfman'1981},\cite{DorfmanFokas'1992}. There exist several different approaches to noncommutative integrability (see, for example \cite{Krichever'1981, Kontsevich'1993, EtingofGelfandRetakh'1997, OlverSokolov'1998, EtingofGelfandRetakh'1998, MasashiToda'2003, RetakhRubtsov'2010}). In this paper we use the approach developed by Mikhailov and Sokolov. In paper \cite{MikhailovSokolov'2000} they suggested to consider Lax operator $L$ of integrable system as an element of some noncommutative associative algebra $\mathcal A$ instead of taking a matrix of a particular size. In this approach we can introduce the equation of the spectral curve as a polynomial relation on $L\in\mathcal A$ of the form $P(L)=0$. Then the passage to the usual Lax matrix is equivalent to represent a quotient of algebra $\mathcal A$ by the ideal generated by the spectral curve $P(L)=0$, followed by the choice of representation of $\mathcal A/(P(L))$. In this paper we show that the idea of integrability can be formulated in universal terms for associative algebra $\mathcal A$ itself, independently of the particular choice of representation. This type of equations should naturally incorporate known classical and quantum integrable systems as well as provide a basis for further generalization.

Classical Lax matrix $L(z)$ with a spectral parameter provides generating functions $\textrm{Tr} (L(z))^k$ for Hamiltonians of the system, which in turn generate a complete set of the first integrals of motion. So, to use all advantages of the powerful Inverse Scattering Method \cite{FadeevTakhtajan'1986, SokolovShabat'1984, MikhailovShabatYamilov'1987, Fokas'1987} in the noncommutative case one have to define trace first. The proper candidate for the role of codomain of traces in the case of associative algebras is the so-called 0-Hochschild homology.
We know that in the case of $M\in Mat(N,\mathbb C)$ the trace of $M^k$ is invariant under conjugation of $M$ by elements of $GL(N,\mathbb C)$, so it seems natural to define the trace as an equivalence class under conjugation with generators of associative algebra\footnote{More precisely, the set of $\textrm{Tr}\,M^k,\;1\leqslant k\leqslant N$ parameterizes generic orbit of the coadjoint action of $GL(N,\mathbb C)$. We can view this orbit as an equivalence class of matrices modulo conjugation which motivates this definition.}. The quotient of the associative algebra modulo these equivalence relations we call a \textit{cyclic space}. The accurate definitions of the corresponding notions are given in Sec. \ref{sec:CyclicSpace}. The cyclic space has no longer a structure of associative algebra, however it appears that it can be endowed with a Lie bracket \cite{Kontsevich'1993, VandenBergh'2008, CrawleyBoevey'2011}.

The latter brings us to the point that Hamiltonians of integrable systems on associative algebras are elements of the cyclic space. To define the corresponding Hamilton equations of motion
\begin{align*}
\forall x\in\mathcal A,\qquad\frac{\textrm d}{\textrm dt}x=\{H,x\}.
\end{align*}
we have to define $\{\_,\_\}$ --- a noncommutative analog of the Poisson bracket. The above equation must define a derivation of $\mathcal A$ and thus $\{\_,\_\}$ should satisfy the Leibnitz rule in the second argument. However, the skew-symmetricity is no longer a requirement, since the two arguments of $\{\_,\_\}$ are the elements of two different spaces. After some additional analysis in Sec. \ref{sec:HamiltonianFlow} we come to conclusion that the proper type of the bracket we need is actually a so-called Loday bracket \cite{Kosmann-Schwarzbach'1996,Loday'1998}.

Loday bracket satisfies Leibnitz rule only in the second argument. The absence of the Leibnitz rule in the first argument naively prevents us from defining the bracket by its action on generators of associative algebra. To overcome this subtlety we extend Van den Bergh approach to construction of a certain subclass of Loday algebras via the double Poisson bracket \cite{VandenBergh'2008}. The double Poission bracket satisfy some forms of the Leibnitz rule w.r.t. both arguments, which makes the definition more constructive.

We present the definition of double Poisson algebras by Van den Bergh in Sec. \ref{sec:DoublePoisson}, this original definition requires the bracket to be skew-symmetric in specific sense, which becomes crucial for derivation of Jacobi Identity. However, for our case we have to abandon the skew-symmetricity condition. It appears that the Jacobi Identity for the resulting Loday algebra can still be satisfied by the bracket presented in Sec. \ref{sec:KontsevichDoubleBracket}. As a result we present the bracket of Loday type which describes the equation of motion for the Kontsevich system.

Finally in Sec. \ref{sec:NonCommutativeLiouvilleIntegrability} we formulate a notion of Liouville integrability for systems on associative algebras and introduce the \textit{Associative} Inverse Scattering Method. With this approach the role of the Lax matrix is played by some element of the associative algebra. The corresponding commuting Hamilton flows are generated by noncommutative traces of powers of the Lax element.

The above notion of noncommutative integrability is an alternative to the one introduced by Kontsevich in paper \cite{Kontsevich'2011}. The relation between these two notions of noncommutative integrability could be an interesting topic for further research.



\section{Kontsevich system}

Let $A=\mathbb C\langle u^{\pm1},v^{\pm1}\rangle$ denotes the associative group algebra over $\mathbb C$ of the free group $G=\langle u,v\rangle$ with two generators. Kontsevich proposed a noncommutative system of ODE's on this algebra
\begin{align}
\left\{\begin{array}{l}
\dfrac{\textrm du}{\textrm dt}=uv-uv^{-1}-v^{-1},\\[10pt]
\dfrac{\textrm dv}{\textrm dt}=-vu+vu^{-1}+u^{-1},
\end{array}\right.
\label{eq:KontsevihSystem}
\end{align}
which admits the following discrete symmetry
\begin{align}
u\rightarrow uvu^{-1},\qquad v\rightarrow u^{-1}+v^{-1}u^{-1}.
\end{align}
The latter can be viewed as a noncommutative analog of B\"acklund transformations. Based on this data Kontsevich conjectured that (\ref{eq:KontsevihSystem}) is integrable.

In paper \cite{EfimovskayaWolf'2012} it was proved that system (\ref{eq:KontsevihSystem}) admits the Lax representation
\begin{align}
\frac{\textrm dL}{\textrm dt}=[L,M]
\end{align}
with the following Lax pair
\begin{align}
L=\left(\begin{array}{cc}
v^{-1}+u&\lambda v+v^{-1}u^{-1}+u^{-1}+1\\
v^{-1}+\frac1\lambda u&v+v^{-1}u^{-1}+u^{-1}+\frac1\lambda
\end{array}\right),
\qquad
M=\left(\begin{array}{cc}
v^{-1}-v+u&\lambda v\\
v^{-1}&u
\end{array}\right).
\label{eq:LaxPair}
\end{align}
This Lax pair gives a rise to an infinite number of Hamiltonians which cover all independent first integrals of (\ref{eq:KontsevihSystem}) as was conjectured in \cite{EfimovskayaWolf'2012}.

\subsection{Classical or commutative counterpart}

Note, that in the commutative (classical) case the equations (\ref{eq:KontsevihSystem}) are Hamilton
\begin{align*}
\frac{\textrm{d}u}{\textrm{d}t}=\{h,u\},\qquad\qquad \frac{\textrm{d}v}{\textrm{d}t}=\{h,v\}
\end{align*}
with respect to the Hamiltonian
\begin{align}
h=u+v+u^{-1}+v^{-1}+u^{-1}v^{-1}
\label{eq:ClassicalKontsevichHamiltonian}
\end{align}
and Poisson bracket defined by
\begin{align}
\{v,u\}=uv.
\label{eq:ClassicalBracket}
\end{align}
This implies that the commutative counterpart of (\ref{eq:KontsevihSystem}) is trivially integrable in the Liouville sense.

Bracket (\ref{eq:ClassicalBracket}) can be transformed to canonical one via change of variables $u=\textrm e^p,\; v=\textrm e^q,$ which gives $\{p,q\}=1$. Then the Hamiltonian acquires the following form
\begin{align}
h=\textrm e^p+\textrm e^{-p}+\textrm e^{q}+\textrm e^{-q}+\textrm e^{-p-q}
\label{eq:ClassicalPQKontsevichHamiltonian}
\end{align}

\subsection{Hamiltonians and first integrals}

Equations (\ref{eq:KontsevihSystem}) preserve the commutator of the underlying free group $\langle u,v\rangle$:
\begin{align}
uvu^{-1}v^{-1}=c.
\label{eq:Casimir}
\end{align}
In particular, when $c$ is central, the group algebra $\mathbb C\langle u^{\pm1},v^{\pm1}\rangle$ turns into a ``quantum group" and this restriction is consistent with equations (\ref{eq:KontsevihSystem}). Later in the text we show that this is a noncommutative analog of the Casimir element \cite{PichereauVandeWeyer'2008}.

On the classical level, we have another integral of motion, namely Hamiltonian (\ref{eq:ClassicalKontsevichHamiltonian}). Consider its noncommutative analog\footnote{Note, that the ordering in the last term doesn't matter, due to the trivial symmetry of equations of motion $u\leftrightarrow v$, $t\rightarrow -t$.}
\begin{align}
h=u+v+u^{-1}+v^{-1}+u^{-1}v^{-1}.
\label{eq:KontsevichFirstHamiltonian}
\end{align}
It is no longer a first integral of equations of motion. Indeed
\begin{align*}
\frac{\textrm d}{\textrm dt}h(t)=u^{-1}-vuv^{-1}+uv-vu+v^{-1}u^{-1}-u^{-1}v^{-1} +u^{-1}v^{-1}u^{-1} -u^{-2}v^{-1}\neq0.
\end{align*}
However, if we consider a matrix representation $\varphi:\;A\rightarrow Mat_{N\times N}(\mathbb C)$ for any $N$ we get
\begin{align}
\frac{\textrm d}{\textrm dt}\textrm{Tr}\,\varphi(h)=0.
\label{eq:htraceintegral}
\end{align}
Following terminology of \cite{MikhailovSokolov'2000, EfimovskayaWolf'2012} we call $h$ a "trace"-integral. Indeed, even more interesting property holds: for any representation $\varphi$ we have
\begin{align}
\frac{\textrm d}{\textrm dt}\textrm{Tr}\,\varphi(h^k)=0.
\label{eq:hktraceintegral}
\end{align}
In other terms for all representations $\varphi(h)(t)$ has adjoint dynamics under (\ref{eq:KontsevihSystem})
\begin{align}
\varphi(h)(t)=g(t)\varphi(h)(0)g^{-1}(t).
\label{eq:hadjointdynamics}
\end{align}

\section{Noncommutative Hamilton equations of motion for Kontsevich system}

In this section we present equations (\ref{eq:KontsevihSystem}) in the Hamilton form with Hamiltonian (\ref{eq:KontsevichFirstHamiltonian}).

In formulas (\ref{eq:htraceintegral}) and (\ref{eq:hktraceintegral}) we pointed out that the noncommutative analog of the classical Hamiltonian is not literally a first integral of equations of motion. Here we develop these ideas without using representations of $A$. This kind of objects was called the ``trace"-integrals in paper \cite{MikhailovSokolov'2000}. This means that the proper space for Hamiltonians is not an associative algebra itself, but certain quotient space of it, which makes the derivative of $h^k$ to be an element of the zero equivalence class. In particular, for this purpose we use the so-called cyclic space $A/[A,A]$ defined in Sec. \ref{sec:CyclicSpace}.

The idea of making the hamiltonians to be the elements of the different space ruins the anticommutativity of the possible noncommutative analog of the Poisson bracket. In Sec. \ref{sec:HamiltonianFlow} we analyze the necessary requirements for this bracket and come to the conclusion that we should utilize the so-called left Loday bracket. However to provide a constructive definition for the particular Loday bracket it is necessary to employ some object which satisfies the Leibnitz identity in both arguments.

In Sect. \ref{sec:DoublePoisson} we remind the original definition of double Poisson bracket by Van den Bergh \cite{VandenBergh'2008}. The so-called double Poisson bracket uses the tensor square of associative algebra as a codomain. Tensor square of associative algebra admits two independent structures of bimodule. This allows one to define the bracket by its action on generators which is not true for Loday brackets. Van den Bergh has shown that double Quasi-Poisson brackets can provide a Loday bracket after a composition with multiplication map.

However, for our purpose we are required to waive the very strict skew-symmetricity condition used by Van den Bergh, so in Sec. \ref{sec:KontsevichDoubleBracket} we define a modified Double Poisson bracket.

We conclude this section with Hamilton equations and Lax formalism for Kontsevich system in Sec. \ref{sec:KontsevichHamiltonians} and its specialization to classical and quantum Integrable Systems in Sec. \ref{sec:KontsevichSpecialization}.

\subsection{Cyclic space}
\label{sec:CyclicSpace}

Take some group $G$ (we don't impose any restrictions on this group throughout Sec. \ref{sec:CyclicSpace}, however we subsequently use the free group $G=\langle u,v\rangle$ with two generators) and consider its group algebra $\mathcal A=\mathbb C[G]$ over $\mathbb C$. It has a structure of an associative algebra over $\mathbb C$.

We can view it as a vector space of dimension $|G|$ (which in the case $G=\langle u,v\rangle$ is clearly infinite) with the base elements exactly the elements of a group $G$. Now we can form another linear space by taking the quotient of $\mathcal A$ as a vector space by the commutant
\begin{align}
[\mathcal A,\mathcal A]=\textrm{span}\{ab-ba\,|\, a,b\in G\}.
\label{eq:CyclicRelations}
\end{align}
Then we immediately get that all cyclic permutations of product of generators $x_i$ are equivalent modulo (\ref{eq:CyclicRelations}), for example $x_ix_jx_k\equiv x_jx_kx_i\equiv x_kx_ix_j\bmod [\mathcal A,\mathcal A]$
but in general $x_ix_jx_k\not\equiv x_ix_kx_j\bmod [\mathcal A,\mathcal A]$. Easy to see that equivalence of all monomials of generators under the cyclic permutations is a complete (which means sufficient) set of additional relations we get after taking the quotient $\mathcal A_{\natural}=\mathcal A/[\mathcal A,\mathcal A]$ and this fact motivates the name.

The resulting linear space $\mathcal A_{\natural}$ is not an associative subalgebra of $\mathcal A$, since the equivalence class of a zero $[\mathcal A,\mathcal A]$ is not an ideal in $A$. However it appears that this linear space $\mathcal A_{\natural}$ can be endowed with symplectic structure \cite{GelfandSmirnov'1994}.

Now, let $\pi:\mathcal A\rightarrow\mathcal A_\natural$ denotes a natural projection map, then (\ref{eq:hktraceintegral}) can be rewritten as
\begin{align*}
\frac{\textrm d}{\textrm dt}\pi(h^k)=0.
\end{align*}
In Sec. \ref{sec:HamiltonianFlow} we show, that it seems reasonable to understand the images $\pi(h^k)\in \mathcal A_\natural$ as Hamiltonians, and the cyclic space $\mathcal A_\natural=\mathcal A/[\mathcal A,\mathcal A]$ as a natural space for them. At the same time $h$ itself plays the same role as Lax matrix in classical integrable systems.

\subsection{Hamilton Flow for Eq. (\ref{eq:KontsevihSystem})}
\label{sec:HamiltonianFlow}

To present equations of motion (\ref{eq:KontsevihSystem}) in the Hamilton form
\begin{align*}
u=\{h,u\}_K,\qquad v=\{h,v\}_K
\end{align*}
with Hamiltonian (\ref{eq:KontsevichFirstHamiltonian}) we must provide a bracket s.t. for each Hamiltonian it defines a derivation of $A=\mathbb C\langle u,v\rangle$:
\begin{align*}
\forall a\in A\quad\frac{\textrm da}{\textrm dt}=\{h,a\}_K
\end{align*}
and thus should satisfy the Leibnitz rule in the second argument. On the other hand we already pointed out that in the case of noncommutative integrable systems the Hamiltonians are not literaly invariant under dynamics (\ref{eq:hadjointdynamics}). The invariant in this case is the image of $\pi(h)$ in the cyclic space, so we should require the bracket to be invariant under the cyclic permutations of monomials of the first argument. Or, equivalently we can say that the first argument of the bracket is actually the element of the cyclic space.

Now, the bracket becomes a function of elements in two different spaces and the exact anticommutativity cannot be imposed. Keeping this in mind, we immediately have a lot of inequivalent forms of Jacobi identity, and to restore the proper ordering we should go back to properties we want to be secured by it. From the point of view of Integrable Systems, the Jacobi identity is used to ensure that the commutator of the vector fields generated by two different Hamiltonians is the vector field corresponding to their commutator:
\begin{align}
\forall H_1,H_2,x\in A\qquad\{H_1,\{H_2,x\}\}-\{H_2,\{H_1,x\}\}=\{\{H_1,H_2\},x\}
\label{eq:JacobiTypeII}
\end{align}
The structure which satisfies this form of Jacobi identity and is not necessary anti-symmetric is called a Loday algebra \cite{Kosmann-Schwarzbach'1996,Loday'1998}. In what follows we refer to (\ref{eq:JacobiTypeII}) as Jacobi Identity of Loday type.

The Loday bracket satisfies the Leibnitz rule only in the second argument, which makes it impossible to define it's action on the whole associative algebra by its action on the generators unless we impose some additional requirement. It appears that the large class of Loday brackets can be constructed by means of the so-called double Quasi-Poisson brackets \cite{VandenBergh'2008}. However, we need a further modification of this structure for our purpose. Namely, we do not assume that the modified double Poisson bracket has to satisfy the Anti-Symmetricity condition (Sec. \ref{sec:DoublePoisson}, property (\ref{it:Doubleskew-symmetricity})). In Sec. \ref{sec:DoublePoisson} we present the Van den Bergh original definition for double Poisson bracket and breifly sketch its relation to Loday brackets.

\subsection{Double Poisson bracket by Van Den Bergh}
\label{sec:DoublePoisson}

In papers \cite{VandenBergh'2008} and \cite{CrawleyBoevey'2011} the general approach to the symplectic structure on noncommutative associative algebras was developed. It appears that the naive definition of the Poisson bracket on associative algebra $\mathcal A$ as a map $\mathcal A\times\mathcal  A\rightarrow\mathcal  A$ which satisfies the Jacobi and Leibnitz identities doesn't provide essential number of nontrivial examples, since in most cases it is restricted to be proportional to the commutator \cite{FarkasLetzter'1998,CrawleyBoevey'2011}. To resolve this issue Van den Bergh replaced the codomain and introduced the so-called double Poisson bracket, which was originally defined as a map
$\db{\_}:\;\mathcal  A\times\mathcal  A\rightarrow\mathcal  A\otimes\mathcal  A$ which satisfies the following properties:
\begin{enumerate}
\item{\textbf{Bilinearity:}} $\db{\_}$ is bilinear, as a consequence by the universal property it extends to the full tensor product
    \begin{align*}
    \db{\_}:\quad\mathcal  A\otimes\mathcal  A\rightarrow\mathcal  A\otimes\mathcal  A.
    \end{align*}
    In the rest of the Sec. \ref{sec:DoublePoisson} we imply this extension when use notation $\db{\_}$, since it is more natural. Altough, in almost all identites we apply double bracket to pure products of the form $a\otimes b$.
    \label{it:UniversalVarphi}
\item{\textbf{Leibnitz Identity:}\footnote{The idea of using the tensor product as a proper codomain for the double Poisson bracket allows us to endow $\mathcal A\otimes\mathcal  A$ with two independent structures of $\mathcal A$-bimodule. Namely, $x\circ_1(a\otimes b)\circ_1 y=xa\otimes by$ and $x\circ_2(a\otimes b)\circ_2y=ay\otimes xb$. This idea allowed Van den Bergh to combine Leibnitz and Jacobi identities without being to restrictive. One can note that Leibnitz identity utilizes different $\mathcal A$-bimodule structures for the first and the second component.}}
    \begin{align*}
    \db{a\otimes bc}=&\db{a\otimes b}(1\otimes c)+(b\otimes1)\db{a\otimes c}\\
    \db{ab\otimes c}=&\db{a\otimes c}(b\otimes1)+ (1\otimes a)\db{b\otimes c}
    \end{align*}
\item{\textbf{Anti-Symmetricity (Strong):}} $\db{a\otimes b}=-\db{b\otimes a}^{\circ}$ where $^{\circ}$ --- denotes the opposite of the tensor product, for pure products $(a\otimes b)^{\circ}=b\otimes a$.
    \label{it:Doubleskew-symmetricity}
\item{\textbf{Jacobi Identity (Type I):}} To write the Jacobi Identity we first introduce some useful notation. The double Poisson bracket from part \ref{it:UniversalVarphi} gives a rise to a map
    \begin{align*}
    R_{m,n}:\;\underbrace{\mathcal A\otimes\mathcal  A\otimes\dots\otimes\mathcal  A}_{k\;\textrm{times}}\rightarrow\underbrace{\mathcal A\otimes\mathcal  A\otimes\dots\otimes\mathcal  A}_{k\;\textrm{times}}
    \end{align*}
    s.t. if $\db{a\otimes b}=\sum_ix_i(a,b)\otimes y_i(a,b)$ then
    \begin{align*}
    R_{m,n}(a_1\otimes\dots\otimes a_k)=\sum_i a_1\otimes\dots\otimes x_i(a_m,a_n)\otimes\dots\otimes y_i(a_m,a_n)\otimes\dots\otimes a_k
    \end{align*}
    In particular, for $k=2$ we have $\db{\_}=R_{12}$. This is a noncommutative analog of $R$-matrix. With this notation Jacobi Identity reads
    \begin{align}
    R_{12}R_{23}+R_{31}R_{12}+R_{23}R_{31}=0
    \label{eq:StrongJacobi}
    \end{align}
    \label{it:StrongJacobi}
\end{enumerate}
There is a natural multiplication map on a tensor product of copies of associative algebra
\begin{align*}
\mu:\quad\mathcal  A\otimes\mathcal  A\rightarrow\mathcal  A,\quad \mu(a\otimes b)=ab.
\end{align*}
Define $\{a,b\}=\mu(\db{a\otimes b}),\,\forall A,B\in\mathcal A$. The following statements are due to Van den Bergh.
\begin{lemma}\cite{VandenBergh'2008} Bracket $\{\_\}$ defined above have the following properties:
\begin{enumerate}[(i)]
\item{Leibnitz identity in the second argument:}
    \begin{align*}
    \{a\otimes bc\}=\{a\otimes b\}c+b\{a\otimes c\}
    \end{align*}
    \label{it:OrdinaryLeibnitz}
\item{Invariance under cyclic permutations of monomials in the first argument:}
    \begin{align*}
    \{ab\otimes c\}=\{ba\otimes c\}
    \end{align*}
    \label{it:OrdinaryCyclic}
\item{skew-symmetricity modulo $[\mathcal A,\mathcal A]$:}
    \begin{align*}
    \{a\otimes b\}\simeq-\{b\otimes a\}\bmod [\mathcal A,\mathcal A]
    \end{align*}
    \label{it:Ordinaryskew-symmetric}
\item As a consequence of part (\ref{it:OrdinaryCyclic}) bracket gives rise to a map $(\mathcal A/[\mathcal A,\mathcal A])\otimes\mathcal  A\rightarrow\mathcal  A$.
\item Moreover, from parts (\ref{it:OrdinaryCyclic})--(\ref{it:Ordinaryskew-symmetric}) we conclude that it induces a map $\mathcal A/[\mathcal A,\mathcal A]\otimes\mathcal  A/[\mathcal A,\mathcal A]\rightarrow\mathcal  A/[\mathcal A,\mathcal A]$.
\end{enumerate}
\label{lemm:VanDenBerrgh}
\end{lemma}

However, the Jacobi Identity (\ref{eq:StrongJacobi}) on $\mathcal A$ itself is neither required, nor guarantee the desired properties. Here we do not discuss modifications of Jacobi Identity presented by Van den Bergh to secure the Jacobi identity of Loday type (\ref{eq:JacobiTypeII}) for the resulting Loday algebra since they heavily rely on the strong version of Anti-Symmetricity condition presented in part (\ref{it:Doubleskew-symmetricity}) of the properties list (Lemma \ref{lemm:VanDenBerrgh}).

Odesskii, Rubtsov, and Sokolov \cite{OdesskiiRubtsovSokolov'2012,OdesskiiRubtsovSokolov'2013}  carried out the complete classification for linear and quadratic double Poisson brackets. In particular, from the necessary requirement ((2.21)--(2.22) in \cite{OdesskiiRubtsovSokolov'2013}) we conclude that there is no double Poisson bracket which has the classical counterpart (\ref{eq:ClassicalBracket}). This leads us to the point that for our purpose we need to weaken some restrictions in the definition of the Poisson bracket. It appears, that Bilinearity and Leibnitz Identity are essential for us, however the Anti-Symmetricity and Jacobi Identity can be substantially modified. In the next section we present an explicit construction for modified double Poisson bracket which helps to present equations (\ref{eq:KontsevihSystem}) in the Hamilton form.

\subsection{Modified Double Poisson bracket for Kontsevitch system}
\label{sec:KontsevichDoubleBracket}

In this section we construct a bracket on associative algebra $A=\mathbb C\langle u^{\pm1},v^{\pm1}\rangle$ which allows us to present equation (\ref{eq:KontsevihSystem}) in the Hamilton form. We define $\{\_\}_K:\;A\times A\rightarrow A$ as a composition of modified double quasi-Poisson bracket and multiplication map $\mu$
\begin{align}
\{a,b\}_K=\mu\left(\db{a\otimes b}_K\right)\
\label{eq:MultipliedKBracket}
\end{align}
Where the modified double quasi-Poisson bracket $\db{\_}_K$ is defined by its action on the generators
\begin{align}
\boxed{
\db{u,v}_K=-vu\otimes 1,\qquad
\db{v,u}_K=uv\otimes1,\qquad \db{u,u}_K=\db{v,v}_K=0.
}
\label{eq:KBracketOnGenerators}
\end{align}
along with the following requirements
\begin{enumerate}
\item{\textbf{Bilinearity:}} $\db{\_}_K$ is bilinear and thus extends to
    \begin{align*}
    \db{\_}_K:\quad A\otimes A\rightarrow A\otimes A\quad
    \end{align*}
    Again, we use further the same notation $\db{\_}_K$ for extension of this bracket to $A\otimes A$ as well as for operation defined on $A\times A$.
    \label{it:KBracketLinearity}
\item{\textbf{Leibnitz Identity:}}
    \begin{subequations}
    \begin{align}
    \db{a\otimes bc}_K=&\db{a\otimes b}_K(1\otimes c)+(b\otimes1)\db{a\otimes c}_K
    \label{eq:KBracketOuterModule}\\
    \db{ab\otimes c}_K=&\db{a\otimes c}_K(b\otimes1)+ (1\otimes a)\db{b\otimes c}_K
    \label{eq:KBracketInnerModule}
    \end{align}
    \label{eq:DoubleLeibnitz}
    \end{subequations}
    \label{it:KBracketLeibnitz}
\end{enumerate}
Properties (\ref{it:KBracketLinearity})--(\ref{it:KBracketLeibnitz}) along with formulae (\ref{eq:KBracketOnGenerators}) define $\db{\_}_K$ completely. Following \cite{VandenBergh'2008} we employ useful notation. Let $x,y\in A$ then we define the components of the bracket of their product via $\left(\db{x,y}_K'\right)_i$ and $\left(\db{x,y}_K''\right)_i$ as below
\begin{align*}
\db{x\otimes y}_K=:\sum_{i}\left(\db{x,y}_K'\right)_i\otimes\left(\db{x,y}_K'' \right)_i.
\end{align*}
In our case the sum is actually redundant. Rewriting (\ref{eq:KBracketOnGenerators}) we immediately get
\begin{align*}
\db{u\otimes v}_K=:&\db{u,v}_K'\otimes\db{u,v}_K''=-vu\otimes1,\\
\db{v\otimes u}_K=:&\db{v,u}_K'\otimes\db{v,u}_K''=uv\otimes1,\\
\db{u\otimes u}_K=&\db{v\otimes v}_K=0,
\end{align*}
and then extend it to $u^{\pm1},v^{\pm1}$ by Leibnitz identity (\ref{eq:DoubleLeibnitz})
\begin{align*}
\db{u^{-1}\otimes v^{-1}}_K=&(v^{-1}\otimes u^{-1})\db{u\otimes v}_K(u^{-1}\otimes v^{-1})=-1\otimes u^{-1}v^{-1},\\
\db{v^{-1}\otimes u^{-1}}_K=&(u^{-1}\otimes v^{-1})\db{v\otimes u}_K(v^{-1}\otimes u^{-1})=1\otimes v^{-1}u^{-1},\\[5pt]
\db{u^{-1}\otimes v}_K=&-(1\otimes u^{-1})\db{u\otimes v}_K(u^{-1}\otimes1)=v\otimes u^{-1},\\
\db{v\otimes u^{-1}}_K=&-(u^{-1}\otimes1)\db{v\otimes u}_K(1\otimes u^{-1})=-v\otimes u^{-1},\\[5pt]
\db{u\otimes v^{-1}}_K=&-(v^{-1}\otimes1)\db{u\otimes v}_K(1\otimes v^{-1})=u\otimes v^{-1},\\
\db{v^{-1}\otimes u}_K=&-(1\otimes v^{-1})\db{v\otimes u}_K(v^{-1}\otimes1)=-u\otimes v^{-1}.
\end{align*}
If $a,b\in A$ are monomials, then we can present them in the following form
\begin{align*}
a=a_1a_2\dots a_k,\quad a_i=u^{\pm1},v^{\pm1},\qquad b=b_1b_2\dots b_m,\quad b_j=u^{\pm1},v^{\pm1}.
\end{align*}
With this notation we have
\begin{align}
\db{a\otimes b}_K=\sum_{i,j}\left(b_1\dots b_{j-1}\db{a_i,b_j}_K'a_{i+1}\dots a_k\right)\otimes\left(a_1a_2\dots a_{i-1}\db{a_i,b_j}_K''b_{j+1}\dots b_m\right)
\label{eq:DoubleBracketExplicit}
\end{align}
and by linearity this formula extends to the full tensor product $A\otimes A$.

Now recall (\ref{eq:MultipliedKBracket}): $\{x,y\}_K=\mu(\db{x\otimes y}_K)$, this defines an operation $\{\_,\_\}_K:\;A\times A\rightarrow A$.
\begin{proposition}
$\{\_,\_\}_K$ satisfies the following properties:
\begin{enumerate}
\item[(1)]{\textbf{Bilinearity:}} $\{\_,\_\}_K$ is bilinear and thus extends to \begin{align*}
    \{\_\}_K:\qquad A\otimes A\rightarrow A;
    \end{align*}
\item[(2a)]{\textbf{Leibnitz Identity in the second argument:}}
    \begin{align*}
    \{a,bc\}_K=\{a,b\}_Kc+b\{a,c\}_K;
    \end{align*}
\item[(2b)]{\textbf{Invariance under cyclic permutations of monomials in the first argument:}}
    \begin{align*}
    \{ab,c\}_K=\{ba,c\}_K;
    \end{align*}
\item[(3)]{\textbf{Skew-symmetricity modulo $[A,A]$:}}
    \begin{align*}
    \{a,b\}_K\equiv\{b,a\}_K\bmod [A,A];
    \end{align*}
\item[(4)]{\textbf{Jacobi Identity}}
    \begin{align*}
    \forall H_1,H_2,x\in A:\qquad\{H_1,\{H_2,x\}_K\}_K-\{H_2,\{H_1,x\}_K\}_K=\{\{H_1,H_2\}_K,x\}_K.
    \end{align*}
\end{enumerate}
\label{prop:KBracketBasicProperties}
\end{proposition}
\begin{proof}
Part (1) is trivial. Part (2a) is provided by the outer bimodule structure of the double bracket. Indeed, apply $\mu$ to both sides of (\ref{eq:KBracketOuterModule}), this reads
\begin{align*}
\{a,bc\}:=&\mu(\db{a\otimes bc}_K)=\mu(\db{a\otimes b}_K(1\otimes c))+\mu((b\otimes1)\db{a\otimes b}_K)=\\
 =&\mu(\db{a\otimes b}_K)c+b\mu(\db{a\otimes c}_K)=\{a,b\}_Kc+b\{a,c\}_K.
\end{align*}
Part (2b) is provided by the inner bimodule structure of the double bracket. Here
\begin{align*}
\{ab,c\}_K:=&\mu(\db{ab\otimes c}_K)=\quad\textrm{by (\ref{eq:KBracketInnerModule})}\\
=&\mu(\db{a\otimes c}_K(b\otimes1))+\mu((1\otimes c)\db{a\otimes b}_K)=\\
=&\mu((1\otimes b)\db{a\otimes c}_K)+\mu(\db{a\otimes b}_K(c\otimes 1))=\\
=&\mu(\db{ba,c}_K)=:\{ba,c\}_K.
\end{align*}

The proof for parts (3) and (4) is quite technical and thus presented in Appendix \ref{sec:AppendixJacobi}. For the sake of brevity for part (4) (Jacobi identity) we present a proof for $\mathbb C\langle u,v\rangle\subset A$ and then sketch the proof for the whole algebra $A$.
\end{proof}

Using Proposition \ref{prop:KBracketBasicProperties} we conclude that $\{\_,\_\}_K$ is well-defined on $A/[A,A]\times A$ and provide a desired Loday bracket.

\subsection{Lax Matrix, Hamiltonians and Casimir elements.}
\label{sec:KontsevichHamiltonians}

Note first that
\begin{align}
\frac{\textrm d}{\textrm dt}h=\{h,h\}_K=[h,v+u^{-1}].
\label{eq:hLaxEquation}
\end{align}
This equation was first presented in \cite{EfimovskayaWolf'2012}. It is a noncommutative analog of the Lax equation, where the role of the first Lax matrix is played by the element $h$, whereas $M=v+u^{-1}$ plays the role of the second Lax matrix. It was claimed in paper \cite{EfimovskayaWolf'2012} that Eq. (\ref{eq:hLaxEquation}) cannot be considered as a Lax equation since this equation doesn't solely define derivatives for generators of associative algebra $A$. However, we point out that it is already enough to define a Loday bracket $\{\_,\_\}_K$ along with $h$ to completely define a derivation of $A$. From this point of view the \textbf{Lax equation (\ref{eq:hLaxEquation}) plays the role of the condition that secures $\pi(h^k)$ to be invariant}.

Actually, even stronger statement is true, namely $\pi(h^k)$ is an infinite chain of commuting Hamiltonians in $A/[A,A]$. This is shown by the following proposition.
\begin{proposition}
For all $N,M>0$ the corresponding hamiltonians $\pi(h^N)$ and $\pi(h^M)$ are in involution: $\{h^N,h^M\}_K\equiv 0\bmod [A,A]$.
\end{proposition}
\begin{proof}
\begin{align}\{h^N,h^M\}_K=\mu(\db{h^N,h^M}_K)= \mu\left(\sum_{j=0}^N\sum_{k=0}^M(h^k\otimes h^j)\db{h,h}_K(h^{N-j-1}\otimes h^{M-k-1})\right).
\label{eq:PropInvolutionLeibnitz}
\end{align}
On the other hand
\begin{align}
\db{h,h}_K=&1\otimes a-h\otimes b+e\otimes 1,\quad\textrm{where}
\label{eq:PropInvolutionDoubleBracket}\\
a=&u^{-1}+v^{-1}-u^{-1}v^{-1}+v^{-1}u^{-1}+u^{-1}v^{-1}u^{-1}+v^{-1}u^{-1}v^{-1} +u^{-1}v^{-1}u^{-1}v^{-1},\nonumber\\
b=&u^{-1}v^{-1},\nonumber\\
e=&uv-vu.\nonumber
\end{align}
Combining (\ref{eq:PropInvolutionLeibnitz}) with (\ref{eq:PropInvolutionDoubleBracket}) we get
\begin{align}
\{h^N,h^M\}_K=&\mu\left(\sum_{j=0}^{N-1}\sum_{k=0}^{M-1}\left(h^{N+k-j-1}\otimes h^jah^{M-k-1}-h^{N+k-j}\otimes h^jbh^{M-k-1}+h^keh^{N-j-1}\otimes h^{M+j-k-1}\right)\right)\nonumber\\
=&\sum_{j=0}^{N-1}\sum_{k=0}^{M-1}\left(h^{N+k-1}ah^{M-k-1}-h^{N+k} bh^{M-k-1}+h^keh^{N+M-k-2}\right)
\label{eq:PropInvolutionMultMap}\\
\equiv&MN(a+e-hb)h^{M+N-2}\bmod [A,A].\nonumber
\end{align}
But, for $N=M=1$ we have (\ref{eq:hLaxEquation}), so using the \textbf{last but one} line of (\ref{eq:PropInvolutionMultMap}) we get
\begin{align*}
(a+e-hb)=[h,v+u^{-1}].
\end{align*}
And finally
\begin{align*}
\{h^N,h^M\}_K\equiv& MN [h,v+u^{-1}]h^{M+N+2}\bmod [A,A]\\
\equiv&0\bmod[A,A].
\end{align*}
\end{proof}
\begin{corollary}
For all $k>0,\;\pi(h^k)$ is integral of a system of equations (\ref{eq:KontsevihSystem}).
\label{prop:hkTraceIntegral}
\end{corollary}
\begin{proof}
$\pi(\frac{\textrm d}{\textrm dt}h^k)=\pi([h^k,v+u^{-1}])=0$.
\end{proof}

Here we should point out the fact that $\pi(h^k)$ as elements of the cyclic space are independent, whereas all $h^k$ are generated by a single element. When we come to the quotient space $A/[A,A]$ it is no longer have a natural multiplication. Or in other words given an equivalence class $\pi(h)\in A/[A,A]$ we cannot pick a proper representative $h\in A$ which generates the whole series. This makes $\pi(h^2)$ to be in principle unidentified by $\pi(h)$. Namely, given the equivalence class $\pi(h)$ we don't have Lax equation (\ref{eq:hLaxEquation}) for each representative of each class in $A$.
\begin{lemma}
The infinite series of commuting hamiltonians $\pi(h^k)$ is linearly independent over $\mathbb C$.
\end{lemma}
\begin{proof}
It is enough to consider highest term in $u$ in $\pi(h^k)$.
\end{proof}

This brings us to conclusion that $h$ is a noncommutative analog of the Lax matrix. In each representation it has adjoint dynamics (\ref{eq:hadjointdynamics}), as well as it generates the infinite series of commuting hamiltonians $\pi(h^k)$. Here $\pi$ is a projection to the cyclic space which can be treated as noncommutative analog of $\textrm{Tr}$.

Now, we left with the task to understand the meaning of Casimir functions. We already pointed out that (\ref{eq:Casimir}) is invariant under dynamics (\ref{eq:KontsevihSystem}). However it appears that even stronger statement is true.
\begin{proposition}
For each $H\in A/[A,A]$ the corresponding Hamilton flow $\{H,\_\}_K$ with respect to bracket (\ref{eq:MultipliedKBracket}) preserves the group commutator $c=uvu^{-1}v^{-1}$
\label{prop:RightCasimir}
\end{proposition}
\begin{proof}
Direct computation shows that
\begin{align*}
\db{u\otimes c}_K=&uv\otimes v^{-1}-uvu\otimes u^{-1}v^{-1}&&=(1\otimes u)r-r(u\otimes1),\qquad\textrm{where}\quad r=uv\otimes u^{-1}v^{-1}\\
\db{v\otimes c}_K=&&&=(1\otimes v)r-r(v\otimes1),
\end{align*}
the same holds for $u^{-1},v^{-1}$. Now we use the induction by Leibnitz identity (\ref{eq:KBracketInnerModule}) to prove that
\begin{align}
\forall a\in A\qquad \db{a\otimes c}_K=(1\otimes a)r-r(a\otimes1)
\label{eq:CasimirDoubleBracket}
\end{align}
Assume that this holds for $a,b$ and prove this for $ab$:
\begin{align*}
\db{ab\otimes c}_K=&\db{a\otimes c}_K(b\otimes 1)+(1\otimes a)\db{b\otimes c}_K=\\
=&\left((1\otimes a)r-r(a\otimes1)\right)(b\otimes1)+(1\otimes a)\left((1\otimes b)r-r(b\otimes1)\right)=\\
=&(1\otimes ab)r-r(ab\otimes1).
\end{align*}
This implies that (\ref{eq:CasimirDoubleBracket}) is valid. Note finally that by applying multiplication map $\mu$ to both sides of (\ref{eq:CasimirDoubleBracket}) we always get zero. This finalizes the proof.
\end{proof}

In other words, it is a right Casimir function for bracket (\ref{eq:MultipliedKBracket}) $\{\_,\_\}_K$. But it is not a left Casimir function, which means that $\pi(c)$ doesn't generate the trivial flow, like it was in the commutative case. Say
\begin{align}
\{c,u\}_K=uvu^{-1}v^{-1}u-u^2vu^{-1}v^{-1}\neq0.
\label{eq:LeftCasimirCounterExample}
\end{align}
However, it satisfies an important property.
\begin{proposition}
For all $H\in A/[A,A]\quad \{H,c\}_K\equiv0\bmod [A,A]$
\end{proposition}
\begin{proof}
Combine Proposition \ref{prop:RightCasimir} and property (3) from Proposition \ref{prop:KBracketBasicProperties}.
\end{proof}
Proposition \ref{prop:KBracketBasicProperties} means that Casimir operator belongs to the center of the Lie Algebra on a cyclic space (the natural space for Hamiltonians).

\textbf{Discussion on generating set for all "trace" integrals.}

Summarizing we can conclude
\begin{corollary}
If $x\in\pi(\mathbb C\langle h,c,c^{-1}\rangle)$ for some
$x\in A/[A,A]$, then $\frac{\textrm d}{\textrm dt}x=\{h,x\}_K\equiv 0\bmod [A,A]$.
\label{Cor:TraceIntegrals}
\end{corollary}
It is an interesting question whether converse to Corollary \ref{Cor:TraceIntegrals} is true. Or, equivalently, whether elements $h,\,c$ and $c^{-1}$ generate a complete set of "trace" integrals. In paper \cite{EfimovskayaWolf'2012} Efimovskaya and Wolf considered possible ``trace"-integrals of equation (\ref{eq:KontsevihSystem}) up to degree 12 and conjectured that they are all generated by the usual traces of powers of Lax matrix (\ref{eq:LaxPair}).

Another experimental comparison shows that $\pi(\textrm{Tr}\,L^k),\,k\leqslant 3$, the images of the traces of powers of the Lax matrix in the cyclic space $A/[A,A]$ generate the linear subspace of the image $\pi(\mathbb C\langle h,c,c^{-1}\rangle)$ of the subalgebra $\mathbb C\langle h,c,c^{-1}\rangle\subset A$ generated by  $h$, Casimir element $c=uvu^{-1}v^{-1}$ and its inverse $c^{-1}$.

We constructed the basis $B_1$ of noncommutative orthogonal polynomials on $A/[A,A]$ spanning ideal $\mathbb C\langle h,c,c^{-1}\rangle$ generated by $h,\,c,\,c^{-1}$. The basis $B_1$ can mapped by $\pi$ and orthogonalized again, which gives us $B_2$ --- the basis of orthogonal polynomials on $\pi(\mathbb C\langle h,c,c^{-1}\rangle)$. Checking whether coefficients in $\lambda$ of $\pi(\textrm{Tr}\,L^k(\lambda))$ are spanned by basis $B_2$ we can check the converse to Corollary \ref{Cor:TraceIntegrals}. We did it up to  $\textrm{Tr}\,L^3(\lambda)$.

\subsection{Specialization to Quantum and Classical Integrable System}
\label{sec:KontsevichSpecialization}

The invariance of the Casimir element $c=uvu^{-1}v^{-1}$ under dynamics (\ref{eq:Casimir}) allows one to construct certain specializations of algebra $A=\mathbb C\langle u^{\pm1},v^{\pm1}\rangle$ consistent with equations of motion. One can impose relation of the form $c=\textrm e^{\textrm i\hbar}\in\mathbb C$. This reduces the algebra to the so-called "quantum group" (the word group is misleading here, although widely accepted). This is the exponential form of the usual Heisenberg algebra $u=\textrm e^p,\;v=\textrm e^q,\;[p,q]=-\textrm i\hbar$. The latter makes the quantum version naturally embedded in the associative case. However the relation between the natural space for noncommutative Hamiltonians, namely the cyclic space, and quantum Hamiltonians is still vague.

Finally, the particular case $c=1$ corresponds to commutative algebra, here the cyclic space coincides with algebra itself and the bracket turns into anticommutative. On the other hand the fact that associative algebra coincides with its cyclic space endows the latter with multiplication, which makes all Hamiltonians $h^k$ algebraically dependent.

\section{General approach to integrable systems on associative algebras}
\label{sec:NonCommutativeLiouvilleIntegrability}

\subsection{Hamilton flow on associative algebras}
\label{sec:GeneralHamiltonFlow}
One-dimensional flow on associative algebra $\mathcal A$ is defined by derivation $\frac{\textrm d}{\textrm dt}$ which satisfies the Leibnitz rule
\begin{align*}
\forall a,b\in\mathcal A,\qquad\frac{\textrm d}{\textrm dt}(ab)=a\left(\frac{\textrm d}{\textrm dt}b\right) +\left(\frac{\textrm d}{\textrm dt}a\right)b
\end{align*}
To present this flow in the Hamilton form
\begin{align*}
\forall a\in\mathcal A\quad\frac{\textrm da}{\textrm dt}=\{h,a\}
\end{align*}
we must provide a bracket s.t. for each hamiltonian $h$ it defines a derivation of $\mathcal A$ and thus should satisfy the Leibnitz rule in the second argument
\begin{align}
\forall a,b,c\in\mathcal A\qquad \{a,bc\}=\{a,b\}c+b\{a,c\}
\label{eq:LodayLeibnitz}
\end{align}

On the other hand we already pointed out that the case of equations (\ref{eq:KontsevihSystem}) the candidate for hamiltonian when treated as element of $\mathcal A$ is not literally invariant under dynamics (\ref{eq:htraceintegral}). What is invariant, is the image of $\pi(h)$ in the cyclic space, so we should require the bracket to be invariant under the cyclic permutations of monomials of the first argument. This would guarantee us that for any $x,y\in\mathcal A$ s.t. the corresponding elements of the cyclic space are the same: $\pi(x)=\pi(y)$ the Hamilton flows $\{x,\_\}$ and $\{y,\_\}$ are also the same. Or equivalently, the above bracket naturally defines a map
\begin{align*}
\{\_,\_\}_I:\;\mathcal A/[\mathcal A,\mathcal A]\times\mathcal A\rightarrow\mathcal A
\end{align*}
where the first argument is an element of the cyclic space --- the natural space for Hamiltonians.

Next, the Hamilton flows in classical integrable systems form a representation of a Poisson Lie algebra of functions. This is secured by Jacobi identity, which means that the commutator of the Hamilton vector fields generated by two different functions $f$ and $g$ is the Hamilton vector field corresponding to their Poisson bracket $\{f,g\}$:
\begin{align}
\qquad\{f,\{g,x\}\}-\{g,\{f,x\}\}=\{\{f,g\},x\}
\end{align}
When we transfer to the case of associative algebras this implicitly means that there exists a Lie bracket on Hamiltonians
\begin{align*}
\{\_,\_\}_{II}:\;\mathcal A/[\mathcal A,\mathcal A]\times\mathcal A/[\mathcal A,\mathcal A]\rightarrow\mathcal A/[\mathcal A,\mathcal A]
\end{align*}
which is skew-symmetric
\begin{align}
\forall a,b\in\mathcal A/[\mathcal A,\mathcal A]\qquad \{a,b\}_{II}=-\{b,a\}_{II}
\label{eq:SecondBracketskew-symmetricity}
\end{align}
and satisfies the Jacobi identity\footnote{One can note that from the existence of representation in $Der\,\mathcal A$ we immediately get that the l.h.s. of the Jacobi identity (\ref{eq:HamiltonianJacobi}) corresponds to trivial element in $Der\,\mathcal A$. Then, if we quotient by the ideal generated by all possible values of the result we eliminate insufficient components of the space for Hamiltonians. This motivates (\ref{eq:HamiltonianJacobi}) to be the very natural condition to impose.}
\begin{align}
\forall a,b,c\in\mathcal A/[\mathcal A,\mathcal A]\qquad \{a,\{b,c\}_{II}\}_{II}+\{b,\{c,a\}_{II}\}_{II} +\{c,\{a,b\}_{II}\}_{II}=0
\label{eq:HamiltonianJacobi}
\end{align}
This means that $\{\_,\_\}_{II}$ is a Lie bracket
enters the analog of Jacobi identity for $\{\_,\_\}_I$
\begin{align*}
\forall f,g\in\mathcal A/[\mathcal A,\mathcal A],\;\forall x\in\mathcal A:\qquad\{f,\{g,x\}_I\}_I-\{g,\{f,x\}_I\}_I=\{\{f,g\}_{II},x\}_{I}.
\end{align*}
where the order of arguments becomes essential. Note that the inner bracket on the r.h.s. is of the different type. The skew-symmetricity (\ref{eq:SecondBracketskew-symmetricity}) requires the following form of skew-symmetricity for the first bracket
\begin{align}
\forall a,b\in\mathcal A\qquad \{\pi(a),b\}_I+\{\pi(b),a\}_I\equiv0\bmod [\mathcal A,\mathcal A].
\label{eq:FirstBracketAntiCommutativity}
\end{align}
Thus bracket $\{\_,\_\}_I$ is the so-called left Loday bracket \cite{Loday'1998}.

Now suppose we have a Hamilton dynamics $\forall f\in \mathcal A\;\frac{\textrm d}{\textrm dt}f=\{h,f\}$.
\begin{definition}
We call the space of Hamiltonians (or "trace"-integrals) $\mathcal H\subset \mathcal A/[\mathcal A,\mathcal A]$ s.t.
\begin{align}
x\in\mathcal H\qquad\Leftrightarrow\qquad\forall x'\in\mathcal A\;\textrm{s.t.}\; \pi(x')=x\quad\frac{\textrm d}{\textrm dt}x'\equiv0\bmod [\mathcal A,\mathcal A]
\label{eq:HamiltonianSpaceDef}
\end{align}
Or, equivalently, one can say that $\{h,x\}_{II}=0$.
\end{definition}
As in the commutative case each Hamiltonian defines a Hamilton flow, s.t. all other Hamiltonians (as elements of $\mathcal A/[\mathcal A,\mathcal A]$) are invariant under this flow.
This can be presented in the following way
\begin{proposition}
The $\mathcal H$ is a maximal commutative Lie subalgebra in $\mathcal A/[\mathcal A,\mathcal A]$ with respect to bracket $\{\_,\_\}_{II}$.
\end{proposition}
\begin{proof}
Since $h\in\mathcal H$, maximality follows directly from definition. Next, if $h_1,h_2\in\mathcal H$ then from (\ref{eq:HamiltonianJacobi}) $\{h_1,h_2\}_{II}\in\mathcal H$.
\end{proof}

\subsection{Casimir functions}

The analog of the classical Casimir functions is the right Casimir of bracket $\{\_,\_\}_I$.
\begin{definition}
We call $c\in\mathcal A$ to be the Casimir element of bracket $\{\_,\_\}_I$ if $\forall a\in\mathcal A/[\mathcal A,\mathcal A]\quad\{a,c\}_I=0$.
\label{def:Casimir}
\end{definition}
The latter implies only that any element in $\mathcal A/[\mathcal A,\mathcal A]$ defines a derivation of $\mathcal A$ which fixes $c$. But $\pi(c)$ doesn't have to define a trivial Hamilton flow, the counterexample was presented in (\ref{eq:LeftCasimirCounterExample}).

However, we can formulate the following
\begin{proposition}
If $c$ is a Casimir in a sense of Def. \ref{def:Casimir}, then its image in the cyclic space $\pi(c)$ necessary belongs to the center of the Lie bracket $\{\_,\_\}_{II}$ on Hamiltonians.
\end{proposition}
\begin{proof} Apply $\pi$ to both sides of (\ref{eq:FirstBracketAntiCommutativity}).
\end{proof}

\subsection{Lax elements}

Suppose the space of Hamiltonians $\mathcal H$ from (\ref{eq:HamiltonianSpaceDef}) can be constructed as an image of the finitely-generated subalgebra $\mathcal L$ of $\mathcal A$ under the noncommutative trace map $\pi$. Say
\begin{align}
\mathcal H=\pi(\mathcal L)\quad\textrm{where}\quad\mathcal L=\mathbb C\langle l_1,\dots,l_n\rangle,\quad l_j\in\mathcal A.
\label{eq:AssociativeInverseScattering}
\end{align}

Without loss of generality we can assume that $h=\pi(l_1)$ is the Hamiltonian of the system which generates the flow corresponding to the equation of motion. Now, for any representation $\varphi:\;\mathcal A\rightarrow Mat_{N\times N}(\mathbb C)$ of $\mathcal A$ in particular we have that
\begin{align*}
\frac{\textrm d}{\textrm dt}\textrm{Tr}\,\varphi(l_j)^k=0,
\end{align*}
or, equivalently,
\begin{align*}
\varphi(l_j)(t)=g_j(t)\varphi(l_j)(t)g_j^{-1}(t).
\end{align*}
Thus $\varphi(l_j)(t)$ satisfies Lax equation. The latter motivates us to call $l_1,\dots,l_n$ universal Lax elements. The property of their images $\varphi(l_j)$ to have an adjoint dynamics doesn't depend on a particular choice of the representation or a spectral curve.

Indeed, (\ref{eq:AssociativeInverseScattering}) is a very strict condition, which is not valid for generic hamiltonian system on $\mathcal A$. In formula (\ref{eq:hLaxEquation}) we presented a particular form of sufficient condition for $l_1$ to be the Lax element of the Hamiltonian hierarchy. Restrict ourselves to the simplest case
\begin{align*}
\mathcal H_1=\pi(\mathcal L)\quad\textrm{where}\quad\mathcal L=\mathbb C\langle L\rangle,\quad L\in\mathcal A
\end{align*}
which corresponds to the Lax matrix without a spectral parameter. Then, assume that there exist such $M\in\mathcal A$ that
\begin{align*}
\frac{\textrm d}{\textrm dt}L=\{\pi(L),L\}_I=[L,M]
\end{align*}
where the commutator is taken in the associative algebra $\mathcal A$. This guarantees that $\pi(L^k)$ generate the commuting Hamilton flows.

\begin{center}
Acknowledgements
\end{center}

I would like to thank Prof. V.~Rubtsov for fruitful discussions, Prof. A.~Levin and Prof. A.~Mironov for useful remarks. I am especially grateful to Prof. V.~Retakh for formulating the problem and useful discussions. The work was partially supported by RFBR grants 12-01-00482 and 12-02-00594.

\appendix

\section{Proof of Proposition \ref{prop:KBracketBasicProperties}}
\label{sec:AppendixJacobi}
In this section we provide a proof of skew-symmetricity of bracket $\{\_,\_\}_K$ proposed in Sec. \ref{sec:KontsevichDoubleBracket} modulo $[A,A]$
\begin{align*}
\{a,b\}_K\equiv\{b,a\}_K\bmod [A,A].
\end{align*}
We also prove Jacobi identity for subalgebra $\mathbb C\langle u,v\rangle\subset A$
\begin{align}
\forall H_1,H_2,x\in\mathbb C\langle u,v\rangle :\qquad\{H_1,\{H_2,x\}_K\}_K-\{H_2,\{H_1,x\}_K\}_K=\{\{H_1,H_2\}_K,x\}_K.
\label{eq:KJacobiAppendix}
\end{align}
Finally we sketch the proof of the Jacobi identity for the whole algebra $A$
\begin{align*}
\forall H_1,H_2,x\in A :\qquad\{H_1,\{H_2,x\}_K\}_K-\{H_2,\{H_1,x\}_K\}_K=\{\{H_1,H_2\}_K,x\}_K
\end{align*}
which appears to be quite technical.

Unlike the double brackets proposed by Van den Bergh we no longer have an Anti-Symmetricity requirement along with Associative Yang-Baxter equation satisfied by double bracket which used to guarantee the corresponding Jacobi identity. However it appears that bracket $\{\_,\_\}_K$ defined in (\ref{eq:KBracketOnGenerators}) (Sec. \ref{sec:KontsevichDoubleBracket}) can still provide a Loday bracket after composition with multiplication map.

It is worth noticing that this proof is by no means natural and doesn't reveal possible internal structure of the proposed generalization of double brackets.


Note that $\db{H_1\otimes H_2\otimes S}_K$ satisfy the Leibnitz rule in the last argument $S$, so it is enough to prove the Jacobi identity for $\db{H_1\otimes H_2\otimes u}_K$ and $\db{H_1\otimes H_2\otimes v}_K$.

Denote
\begin{align*}
H_1=a_1a_2a_3\dots a_N,\qquad H_2=b_1b_2b_3\dots b_M,\qquad a_i,b_j\in\{u,v\}
\end{align*}
Then we have
\begin{align}
\{H_1,H_2\}=&\sum_{\begin{array}{c}k=0\\a_{k+1}=u\end{array}}^{N-1} \sum_{\begin{array}{c}l=0\\b_{l+1}=v\end{array}}^{M-1}(b_1\dots b_l)(-vu)(a_{k+2}\dots a_k)(b_{l+2}\dots b_N)+\nonumber\\
&+\sum_{\begin{array}{c}k=0\\a_{k+1}=v\end{array}}^{N-1} \sum_{\begin{array}{c}l=0\\ b_{l+1}=u\end{array}}^{M-1}(b_1\dots b_l)(uv)(a_{k+2}\dots a_k)(b_{l+2}\dots b_N).
\label{eq:ProofH1H2}
\end{align}
Hereinafter we employ the following short notation for the "cyclic" product
\begin{align*}
(a_i\dots a_j)=\left\{\begin{array}{cl}
a_ia_{i+1}a_{i+2}\dots a_j,&j\geqslant i,\\
a_ia_{i+2}\dots a_Na_1a_2\dots a_j,&j<i-1,\\
1,&j=i-1.
\end{array}\right.
\end{align*}
Although the last part looks unnatural it is very convenient modification for our case, since we never encounter a complete cyclic permutation of monomial.

Now using (\ref{eq:ProofH1H2}) we get
\begingroup
\allowdisplaybreaks
\begin{subequations}
\begin{align}
\{\{H_1,H_2\}_K,u\}_K=&\nonumber\\
=&uv\sum_{\begin{array}{c}k=0\\a_{k+1}=u\end{array}}^{N-1} \sum_{\begin{array}{c}l=0\\b_{l+1}=v\end{array}}^{M-1} \sum_{\begin{array}{c}m=0\\m\neq l\\b_{m+1}=v\end{array}}^{M-1}(b_{m+2}\dots b_l)(-vu)(a_{k+2}\dots a_k)(b_{l+2}\dots b_m)-
\label{eq:ProofPart18}\\
&-uvu\sum_{\begin{array}{c}k=0\\a_{k+1}=u\end{array}}^{N-1} \sum_{\begin{array}{c}l=0\\b_{l+1}=v\end{array}}^{M-1}(a_{k+2}\dots a_k) (b_{l+2}\dots b_l)+
\label{eq:ProofPart19}\\
&+uv\sum_{\begin{array}{c}k=0\\a_{k+1}=u\end{array}}^{N-1} \sum_{\begin{array}{c}l=0\\b_{l+1}=v\end{array}}^{M-1} \sum_{\begin{array}{c}n=0\\a_{n+1}=v\end{array}}^{N-1} (a_{n+2}\dots a_k)(b_{l+2}\dots b_l)(-vu)(a_{k+2}\dots a_n)+
\label{eq:ProofPart20}\\
&+uv\sum_{\begin{array}{c}k=0\\a_{k+1}=v\end{array}}^{N-1} \sum_{\begin{array}{c}l=0\\b_{l+1}=u\end{array}}^{M-1} \sum_{\begin{array}{c}m=0\\b_{m+1}=v\end{array}}^{M-1} (b_{m+2}\dots b_l)(uv) (a_{k+2}\dots a_k)(b_{l+2}\dots b_m)+
\label{eq:ProofPart21}\\
&+uv\sum_{\begin{array}{c}k=0\\a_{k+1}=v\end{array}}^{N-1} \sum_{\begin{array}{c}l=0\\b_{l+1}=u\end{array}}^{M-1}(a_{k+2}\dots a_k) (b_{l+2}\dots b_l)u+
\label{eq:ProofPart22}\\
&+uv\sum_{\begin{array}{c}k=0\\a_{k+1}=v\end{array}}^{N-1} \sum_{\begin{array}{c}l=0\\b_{l+1}=u\end{array}}^{M-1} \sum_{\begin{array}{c}n=0\\n\neq k\\a_{n+1}=v\end{array}}^{N-1}(a_{n+2}\dots a_k) (b_{l+2}\dots b_l)(uv)(a_{k+2}\dots a_n).
\label{eq:ProofPart23}
\end{align}
\label{eq:ProofBH1H2u}
\end{subequations}
Next, the second term on the l.h.s. of Jacobi identity of type (\ref{eq:KJacobiAppendix}) reads here
\begin{subequations}
\begin{align}
\{H_2,\{H_1,u\}_K\}_K=&\nonumber\\
=&uv\sum_{\begin{array}{c}l=0\\b_{l+1}=v\end{array}}^{M-1} \sum_{\begin{array}{c}k=0\\a_{k+1}=v\end{array}}^{N-1} (b_{l+2}\dots b_l)v (a_{k+2}\dots a_k)-
\label{eq:ProofPart24}\\
&-uvu\sum_{\begin{array}{c}l=0\\b_{l+1}=u\end{array}}^{M-1} \sum_{\begin{array}{c}k=0\\a_{k+1}=v\end{array}}^{N-1} (b_{l+2}\dots b_l) (a_{k+2}\dots a_k)+
\label{eq:ProofPart25}\\
&+uv\sum_{\begin{array}{c}l=0\\b_{l+1}=u\end{array}}^{M-1} \sum_{\begin{array}{c}k=0\\a_{k+1}=v\end{array}}^{N-1} \sum_{\begin{array}{c}n=0\\n\neq k\\ a_{n+1}=v\end{array}} (a_{k+2}\dots a_n) (-vu)(b_{l+2}\dots b_l)(a_{n+2}\dots a_k)+
\label{eq:ProofPart26}\\
&+uv\sum_{\begin{array}{c}l=0\\b_{l+1}=v\end{array}}^{M-1} \sum_{\begin{array}{c}k=0\\a_{k+1}=v\end{array}}^{N-1} \sum_{\begin{array}{c}n=0\\a_{n+1}=u\end{array}}^{N-1} (a_{k+2}\dots a_n)(uv) (b_{l+2}\dots b_l)(a_{n+2}\dots a_k).
\label{eq:ProofPart27}
\end{align}
\label{eq:ProofH2H1u}
\end{subequations}
For the last term we have
\begin{subequations}
\begin{align}
\{H_1,\{H_2,u\}\}=&\nonumber\\
=&uv\sum_{\begin{array}{c}k=0\\a_{k+1}=v\end{array}}^{N-1} \sum_{\begin{array}{c}l=0\\b_{l+1}=v\end{array}}^M (a_{k+2}\dots a_k)v(b_{l+1}\dots b_l)-
\label{eq:ProofPart28}\\
&-uvu\sum_{\begin{array}{c}k=0\\a_{k+1}=u\end{array}}^{N-1} \sum_{\begin{array}{c}l=0\\b_{l+1}=v\end{array}}^{M-1} (a_{k+2}\dots a_k) (b_{l+2}\dots b_l)+
\label{eq:ProofPart29}\\
&+uv\sum_{\begin{array}{c}k=0\\a_{k+1}=u\end{array}}^{N-1} \sum_{\begin{array}{c}l=0\\b_{l+1}=v\end{array}}^{M-1} \sum_{\begin{array}{c}m=0\\m\neq l\\b_{m+1}=v\end{array}}^{M-1} (b_{l+2}\dots b_m)(-vu)(a_{k+2}\dots a_k)(b_{m+2}\dots b_l)+
\label{eq:ProofPart30}\\
&+uv\sum_{\begin{array}{c}k=0\\a_{k+1}=v\end{array}}^{N-1} \sum_{\begin{array}{c}l=0\\b_{l+1}=u\end{array}}^{M-1} \sum_{\begin{array}{c}m=0\\b_{m+1}=v\end{array}}^{M-1}(b_{m+2}\dots b_l)(uv) (a_{k+2}\dots a_k) (b_{l+2}\dots b_m).
\label{eq:ProofPart31}
\end{align}
\label{eq:ProofH1H2u}
\end{subequations}
\endgroup

Combining (\ref{eq:ProofBH1H2u}), (\ref{eq:ProofH2H1u}), and (\ref{eq:ProofH1H2u}) we get the Jacobi Identity. Indeed, first we note that (\ref{eq:ProofPart18}) cancels with (\ref{eq:ProofPart30}). The same happens for pairs (\ref{eq:ProofPart21}), (\ref{eq:ProofPart31}) and (\ref{eq:ProofPart19}), (\ref{eq:ProofPart29}). So we left with 8 terms.

Next, we note that (\ref{eq:ProofPart22}) can be absorbed in (\ref{eq:ProofPart23}) if we allow $n=k$. The same happens with (\ref{eq:ProofPart25}), which is absorbed in (\ref{eq:ProofPart26}). So for now there are 6 terms left which cancel due to less trivial reason.
\begin{align*}
(\ref{eq:ProofPart20}):&\qquad -u(\mathop{a}^{\displaystyle\overset v\downarrow}\mathop a\mathop a\mathop a)(\mathop b\mathop b\mathop b\mathop{b}^{\displaystyle\overset v\downarrow})(\mathop{a}^{\displaystyle\overset u\downarrow}\mathop a\mathop a)\\
(\ref{eq:ProofPart22})+(\ref{eq:ProofPart23}):&\qquad u(\mathop{a}^{\displaystyle\overset v\downarrow}\mathop a\mathop a\mathop a)(\mathop b\mathop b\mathop b\mathop{b}^{\displaystyle\overset u\downarrow})(\mathop{a}^{\displaystyle\overset v\downarrow}\mathop a\mathop a)\\
(\ref{eq:ProofPart24}):&\qquad u(\mathop{b}^{\displaystyle\overset v\downarrow}\mathop b\mathop b\mathop b)(\mathop{a}^{\displaystyle\overset v\downarrow}\mathop a\mathop a\mathop a)\\
(\ref{eq:ProofPart25})+(\ref{eq:ProofPart26}):&\qquad -u(\mathop{a}^{\displaystyle\overset v\downarrow}\mathop a\mathop a\mathop a^{\displaystyle\overset v\downarrow})(\mathop{b}^{\displaystyle\overset u\downarrow}\mathop b\mathop b\mathop b)(\mathop a\mathop a\mathop a)\\
(\ref{eq:ProofPart27}):&\qquad u(\mathop{a}^{\displaystyle\overset v\downarrow}\mathop a\mathop a\mathop a^{\displaystyle\overset u\downarrow})(\mathop{b}^{\displaystyle\overset v\downarrow}\mathop b\mathop b\mathop b)(\mathop a\mathop a\mathop a)\\
(\ref{eq:ProofPart28}):&\qquad -u(\mathop{a}^{\displaystyle\overset v\downarrow}\mathop a\mathop a\mathop a)(\mathop{b}^{\displaystyle\overset v\downarrow}\mathop b\mathop b\mathop b)
\end{align*}

All remaining terms have the similar structure. They start with $u$ on the left followed by the cyclic permutation of $H_1$, with cyclic permutation of $H_2$ inserted inside. Our goal is to prove that each particular monomial $x=uc_1c_2\dots c_{M+N}$ will enter the answer with coefficient $0$. Of course in general there are multiple ways of presenting each monomial in the form described above for a given $H_1$ and $H_2$.

Each coefficient in the Jacobi identity comes as a weighted (by the corresponding coefficients) sum over all presentations of the particular monomial in the forms (\ref{eq:ProofPart20}), (\ref{eq:ProofPart22}), (\ref{eq:ProofPart23}), (\ref{eq:ProofPart24}), (\ref{eq:ProofPart25}), (\ref{eq:ProofPart26}), (\ref{eq:ProofPart27}), or (\ref{eq:ProofPart28}). Moreover, in general each monomial can be presented in the same form listed above several times, depending on which part of $u^{-1}x=c_1c_2\dots c_{M+N}$ we consider as a cyclic permutation of $H_2$ and which as $H_1$. Say we treat $u^{-1}x=(c_1c_2\dots c_i)(c_{i+1}\dots c_{i+M})(c_{i+M+1}\dots c_{N+M})$ so that $c_{i+1}\dots c_{i+M}\equiv H_2\bmod [A,A]$ and $c_1\dots c_ic_{i+M+1}\dots c_{N+M}\equiv H_1\bmod [A,A]$. Then if $c_{i+M+1}=c_{i+1}$ we can shift $i\rightarrow i+1$ to get another presentation of the same monomial. Analogously if $c_{i+M-1}=c_{i-1}$ we can make an opposite shift $i\rightarrow i-1$.

For a given monomial we can fix the value of $i$, then given a quadruple $c_i,c_{i+1},c_{i+M},c_{i+M+1}$ one can determine how many sums can generate such expression. Taking a sum over all possible $i$ we get the coefficient of the monomial $uc_1\dots c_{N+M}$ in the Jacobi identity. Essentially we sum over quadruples $c_i,c_{i+1},c_{i+M},c_{i+M+1}$. Moreover since this quadruples are related for different $i$ it puts certain restrictions on a sum.

Namely, if for $i$ we have quadruple is $c_{i},c_{i+1},c_{i+M},c_{i+M+1}$, than the quadruple for $i+1$ is of the form $c_{i+1},*,c_{i+M+1},*$. The similar applies to $i-1$. We construct an oriented graph, where the vertices correspond to all possible quadruples and an oriented edge connects a pair of quadruples if and only if there exists a monomial $c_1\dots c_{N+M}$ s.t. the first quadruple corresponds to $i$ while the second to $i+1$. We also introduce triples and denote them as quadruple with an empty set as one of the elements. This implies that we reached any end of monomial.

There are $2^4=16$ quadruples involved along with 16 triples of the form $\emptyset***$ and $***\emptyset$. We connect the quadruple/triple with boldface $\textbf0$ if an only if there is no quadruples/triples which can be constructed by the corresponding shift of $i$.

Say, if we take $u,u,v,v$ as a quadruple, then both shifts $i\rightarrow i\pm1$ have no corresponding presentations so we have
\begin{align*}
\textbf0\rightarrow uuvv\rightarrow\textbf0\\
\textbf0\rightarrow uvvu\rightarrow\textbf0\\
\textbf0\rightarrow vuuv\rightarrow\textbf0\\
\textbf0\rightarrow vvuu\rightarrow\textbf0
\end{align*}
as a separate connected components of the graph. For the main connected component we have:
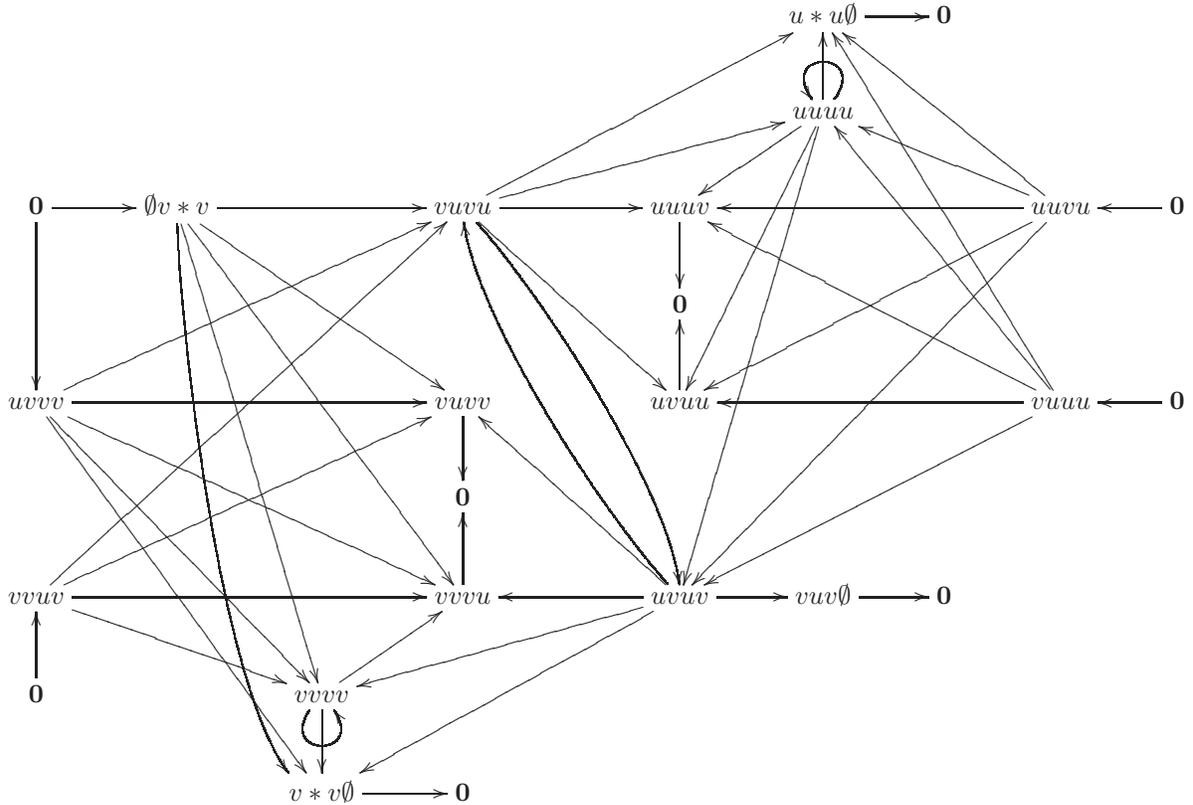
\begin{align}
\xymatrix{
&&&&&&u*u\emptyset\ar[r]&\textbf0\\
&&&&&&uuuu\ar@(ur,ul)\ar[dl]\ar[dddl]\ar[dddddl]\ar[u]\\
\textbf0\ar[r]\ar[dd]&\emptyset v*v\ar[rr]\ar[ddrr]\ar[ddddrr]\ar[dddddr]\ar@(d,l)[ddddddr]&&vuvu\ar[uurrr]\ar[urrr]\ar[rr]\ar[rrdd] \ar@(dr,u)[rrdddd] &&uuuv\ar[d]&&&uuvu\ar[ull]\ar[lll]\ar[ddlll] \ar[llldddd]\ar[uull]&\textbf0\ar[l]\\
&&&&&\textbf 0\\
uvvv\ar[uurrr]\ar[rrr]\ar[ddrrr]\ar[dddrr]\ar[ddddrr]&&&vuvv\ar[d]&&uvuu\ar[u]& &&vuuu\ar[uuull]\ar[llluu]\ar[lll]\ar[ddlll]\ar[uuuull]&\textbf0\ar[l]\\
&&&\textbf0\\
vvuv\ar[uuuurrr]\ar[uurrr]\ar[rrr]\ar[drr]&&&vvvu\ar[u]&& uvuv\ar[ll]\ar[lluu]\ar@(ul,d)[lluuuu]\ar[dlll]\ar[r]\ar[ddlll] &\ar[r]vuv\emptyset&\textbf0\\
\textbf0\ar[u]&&vvvv\ar@(dl,dr)[]\ar[ur]\ar[d]\\
&&v*v\emptyset\ar[r]&\textbf0
}
\label{eq:generalquadruples}
\end{align}
Here by $u*u\emptyset$ we combined two identical vertices $uvu\emptyset$ and $uuu\emptyset$. The same applies to $\emptyset v*v$. To simplify the graph we omitted $\textbf0\rightarrow u*u$, since we do not need this part further.

Now assume the l.h.s. of Jacobi identity (\ref{eq:KJacobiAppendix}) is nonzero. Then it has at least one nonzero monomial, which we denote by $\alpha uc_1\dots c_{N+M}$ where $\alpha\in\mathbb C$ and $c_1,\dots,c_{N+M}\in\{u,v\}$. Now take the smallest value of $i$ s.t. any of the mentioned six sums contribute to the l.h.s. Consider subsequent $i+1$, $i+2$ and so on until we reach the last value of $i+k$ for which at least one sum contributes to the l.h.s. In terms of quadruples $c_{i}c_{i+1}c_{i+M}c_{i+M+1}$ this corresponds to some path on graph $(\ref{eq:generalquadruples})$ from $\textbf0$ to $\textbf0$.

The coefficient $\alpha$ constructed as a sum over contributions corresponding to different values $i\dots i+k$. Where each contribution is in one-to-one correspondence to the quadruple. Below we calculate a contribution for a given value of $i$ of any quadruple
\begin{enumerate}
\item{$uuuu$} doesn't enter a single sum in remaining parts, so the contributing coefficient for a given value of $i$ is 0.
\item{$vvvv$} The same argument as for $uuuu$. The coefficient is 0.
\item{$\emptyset v*v$} For $\emptyset v*v$ we have the corresponding entry in (\ref{eq:ProofPart24}) only. So the coefficient is 1.
\item{$u*u\emptyset$} For $uvu\emptyset$ we have (\ref{eq:ProofPart22})+(\ref{eq:ProofPart23}), (\ref{eq:ProofPart27}), and (\ref{eq:ProofPart28}), so the coefficient is 1. For $uuu\emptyset$ we have (\ref{eq:ProofPart22})+(\ref{eq:ProofPart23}) only, so the coefficient is also 1. Thus we can combine them into $u*u\emptyset$ on the graph.
\item{$vuvu$} enters (\ref{eq:ProofPart20}), (\ref{eq:ProofPart25})+(\ref{eq:ProofPart26}). So the coefficient is -2.
\item{$uvuv$} enters (\ref{eq:ProofPart23}) and  (\ref{eq:ProofPart27}), so the coefficient is +2.
\item{$v*v\emptyset$} enters (\ref{eq:ProofPart28}) only, so coefficient is -1.
\item{$uvvv$} enters (\ref{eq:ProofPart27}) only, so coefficient is 1.
\item{$vvuv$} enters (\ref{eq:ProofPart23}) only, so coefficient is 1.
\item{$vuvv$} enters (\ref{eq:ProofPart25})+(\ref{eq:ProofPart26}) only, so coefficient is -1.
\item{$vvvu$} enters (\ref{eq:ProofPart20}) only, so coefficient is -1.
\item{$uuuv$} enters (\ref{eq:ProofPart23}) only, so coefficient is 1.
\item{$uvuu$} enters (\ref{eq:ProofPart27}) only, so coefficient is 1.
\item{$uuvu$} enters (\ref{eq:ProofPart20}) only, so coefficient is -1.
\item{$vuuu$} enters (\ref{eq:ProofPart25})+(\ref{eq:ProofPart26}) only, so coefficient is -1.
\end{enumerate}

Now, we claim that for any closed path in (\ref{eq:generalquadruples}) the sum of the corresponding contributions is 0. Which is enough to prove that $\alpha=0$. Indeed, replace the quadruples with the corresponding coefficients:
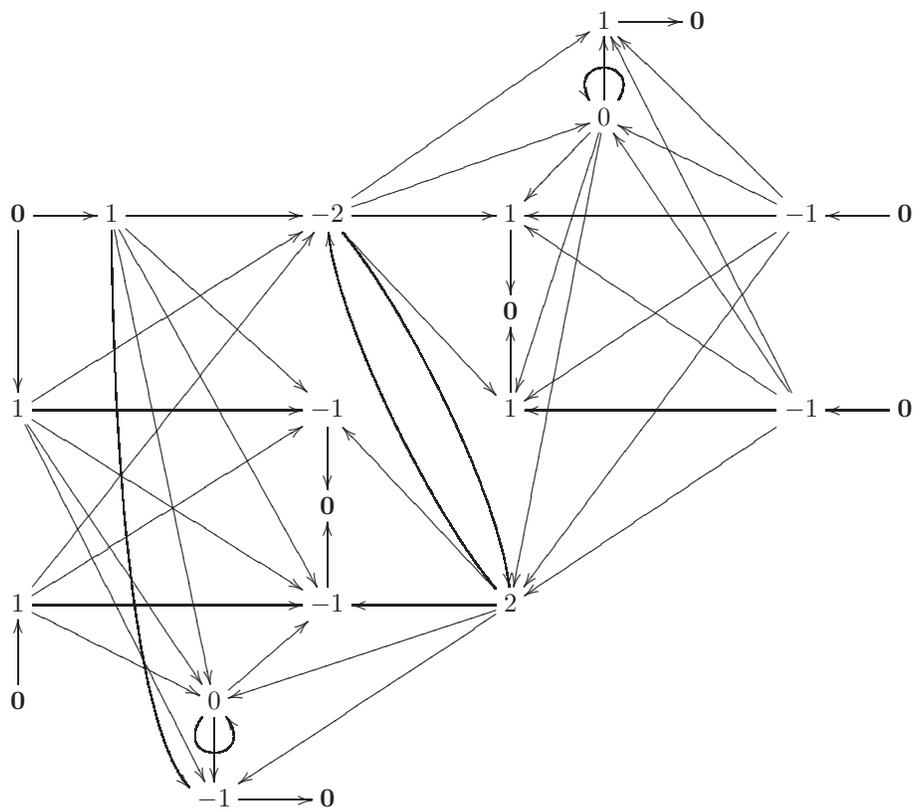
\begin{align*}
\xymatrix{
&&&&&&1\ar[r]&\textbf0\\
&&&&&&0\ar@(ur,ul)\ar[dl]\ar[dddl]\ar[dddddl]\ar[u]\\
\textbf0\ar[r]\ar[dd]& 1\ar[rr]\ar[ddrr]\ar[ddddrr]\ar[dddddr]\ar@(d,l)[ddddddr]&& -2\ar[uurrr]\ar[urrr]\ar[rr]\ar[rrdd] \ar@(dr,u)[rrdddd] &&1\ar[d]&&&-1\ar[ull]\ar[lll]\ar[ddlll] \ar[llldddd]\ar[uull]&\textbf0\ar[l]\\
&&&&&\textbf 0\\
1\ar[uurrr]\ar[rrr]\ar[ddrrr]\ar[dddrr]\ar[ddddrr]&&&-1\ar[d]&&1\ar[u]& &&-1\ar[uuull]\ar[llluu]\ar[lll]\ar[ddlll]\ar[uuuull]&\textbf0\ar[l]\\
&&&\textbf0\\
1\ar[uuuurrr]\ar[uurrr]\ar[rrr]\ar[drr]&&&-1\ar[u]&& 2\ar[ll]\ar[lluu]\ar@(ul,d)[lluuuu]\ar[dlll]\ar[ddlll]\\
\textbf0\ar[u]&&0\ar@(dl,dr)[]\ar[ur]\ar[d]\\
&&-1\ar[r]&\textbf0
}
\end{align*}
One can note that vertices with the zero coefficient can be eliminated by the following rule
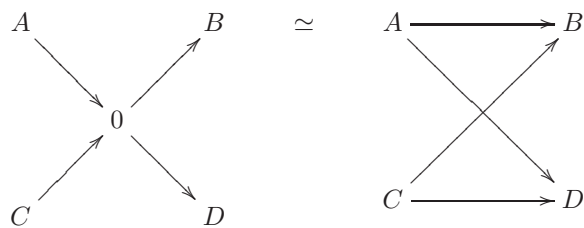
\begin{align*}
\xymatrix{
A\ar[dr]&&B\\
&0\ar[ur]\ar[dr]\\
C\ar[ur]&&D
}\qquad\simeq\qquad
\xymatrix{
A\ar[rr]\ar[ddrr]&&B\\
\\
C\ar[rr]\ar[uurr]&&D
}
\end{align*}
So we left with
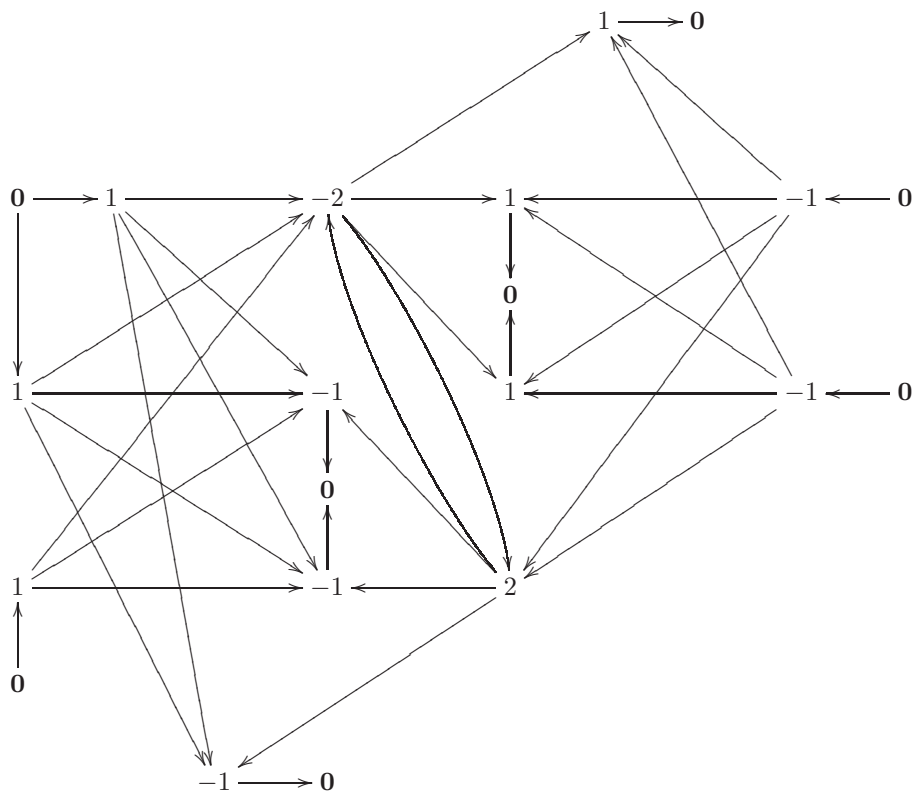
\begin{align*}
\xymatrix{
&&&&&&1\ar[r]&\textbf0\\
\\
\textbf0\ar[r]\ar[dd]& 1\ar[rr]\ar[ddrr]\ar[ddddrr]\ar[ddddddr]&& -2\ar[uurrr]\ar[rr]\ar[rrdd] \ar@(dr,u)[rrdddd] &&1\ar[d]&&&-1\ar[lll]\ar[ddlll] \ar[llldddd]\ar[uull]&\textbf0\ar[l]\\
&&&&&\textbf 0\\
1\ar[uurrr]\ar[rrr]\ar[ddrrr]\ar[ddddrr]&&&-1\ar[d]&&1\ar[u]& &&-1\ar[llluu]\ar[lll]\ar[ddlll]\ar[uuuull]&\textbf0\ar[l]\\
&&&\textbf0\\
1\ar[uuuurrr]\ar[uurrr]\ar[rrr]&&&-1\ar[u]&& 2\ar[ll]\ar[lluu]\ar@(ul,d)[lluuuu]\ar[ddlll]\\
\textbf0\ar[u]\\
&&-1\ar[r]&\textbf0
}
\end{align*}
One can see that in the graph above there is no way from bold $\textbf0$ to another $\textbf0$ without sum of coefficients to be nonzero.

The latter is equivalent to the statement that we can introduce a function $f:V\rightarrow\mathbb Z$ on vertexes of the graph which is a sum of all coefficients on the way from 0 to the vertex. Then for $f$ we have the following table
\begin{align}
\begin{array}{c|ccc}
\textrm{value of }f&+1&0&-1\\
\hline\\
\textrm{list of quadruples}&vvuv&vvuu&vuuu\\
&uvvv&vuuv&vuvu\\
&vvvv&uvvu&uuuu\\
&uvuv&uuvv&uuvu\\
&&vvvu\\
&&vuvv\\
&&uvuu\\
&&uuuv
\end{array}
\label{eq:ProofPotential}
\end{align}
One can note, that $f$ is nonzero iff the second and forth element coincide. If we assign the grading $\deg u=-1,\,\deg v=1$ then the value of $f$ is the gedree of the second element in the quadruple whenever it coincides with the forth element.

Now our goal is to prove that
\begin{align}
\{H_1,H_2\}+\{H_2,H_1\}\equiv 0\bmod [A,A].
\label{eq:AppAntisymmetricity}
\end{align}
Denote
\begin{align*}
H_1=a_1a_2a_3\dots a_N,\qquad H_2=b_1b_2b_3\dots b_M,\qquad a_i,b_j\in\{u^{\pm},v^{\pm}\}
\end{align*}
Then we have

\begingroup
\allowdisplaybreaks
\begin{subequations}
\begin{align}
\{H_1,H_2\}=&-\sum_{\begin{array}{c}k=0\\a_{k+1}=u\end{array}}^{N-1} \sum_{\begin{array}{c}l=0\\b_{l+1}=v\end{array}}^{M-1}(b_1\dots b_l)(vu)(a_{k+2}\dots a_k)(b_{l+2}\dots b_M)+\label{eq:ProofPart40}\\
&+\sum_{\begin{array}{c}k=0\\a_{k+1}=u\end{array}}^{N-1} \sum_{\begin{array}{c}l=0\\b_{l+1}=v^{-1}\end{array}}^{M-1}(b_1\dots b_l)u(a_{k+2}\dots a_k)v^{-1}(b_{l+2}\dots b_M)+\label{eq:ProofPart41}\\
&+\sum_{\begin{array}{c}k=0\\a_{k+1}=u^{-1}\end{array}}^{N-1} \sum_{\begin{array}{c}l=0\\b_{l+1}=v\end{array}}^{M-1}(b_1\dots b_l)v(a_{k+2}\dots a_k)u^{-1}(b_{l+2}\dots b_M)-\label{eq:ProofPart42}\\
&-\sum_{\begin{array}{c}k=0\\a_{k+1}=u^{-1}\end{array}}^{N-1} \sum_{\begin{array}{c}l=0\\b_{l+1}=v^{-1}\end{array}}^{M-1}(b_1\dots b_l)(a_{k+2}\dots a_k)u^{-1}v^{-1}(b_{l+2}\dots b_M)+\label{eq:ProofPart43}\\
&+\sum_{\begin{array}{c}k=0\\a_{k+1}=v\end{array}}^{N-1} \sum_{\begin{array}{c}l=0\\b_{l+1}=u\end{array}}^{M-1}(b_1\dots b_l)uv(a_{k+2}\dots a_k)(b_{l+2}\dots b_M)-\label{eq:ProofPart44}\\
&-\sum_{\begin{array}{c}k=0\\a_{k+1}=v\end{array}}^{N-1} \sum_{\begin{array}{c}l=0\\b_{l+1}=u^{-1}\end{array}}^{M-1}(b_1\dots b_l)v(a_{k+2}\dots a_k)u^{-1}(b_{l+2}\dots b_M)-\label{eq:ProofPart45}\\
&-\sum_{\begin{array}{c}k=0\\a_{k+1}=v^{-1}\end{array}}^{N-1} \sum_{\begin{array}{c}l=0\\b_{l+1}=u\end{array}}^{M-1}(b_1\dots b_l)u(a_{k+2}\dots a_k)v^{-1}(b_{l+2}\dots b_M)+\label{eq:ProofPart46}\\
&+\sum_{\begin{array}{c}k=0\\a_{k+1}=v^{-1}\end{array}}^{N-1} \sum_{\begin{array}{c}l=0\\b_{l+1}=u^{-1}\end{array}}^{M-1}(b_1\dots b_l)(a_{k+2}\dots a_k)v^{-1}u^{-1}(b_{l+2}\dots b_M)\label{eq:ProofPart47}
\end{align}
\label{eq:CompleteH1H2}
\end{subequations}
All sums above again represent the cyclic permutations of $H_1$ inserted into $H_1$. As before we can represent them schematically
\begin{align*}
(\ref{eq:ProofPart40}):&\qquad-(\mathop{b}\mathop{b} \mathop{b}^{\displaystyle\overset v\downarrow}) (\mathop{a}^{\displaystyle\overset u\downarrow} \mathop{a}\mathop{a}\mathop{a}) (\mathop{b}\mathop{b}\mathop{b})\\
(\ref{eq:ProofPart41}):&\qquad (\mathop{b}\mathop{b} \mathop{b}) (\mathop{a}^{\displaystyle\overset u\downarrow} \mathop{a}\mathop{a}\mathop{a}) (\mathop{b}^{\displaystyle\overset {v^{-1}}\downarrow}\mathop{b}\mathop{b})\\
(\ref{eq:ProofPart42}):&\qquad (\mathop{b}\mathop{b} \mathop{b}^{\displaystyle\overset v\downarrow}) (\mathop{a} \mathop{a}\mathop{a}\mathop{a}^{\displaystyle\overset {u^{-1}}\downarrow}) (\mathop{b}\mathop{b}\mathop{b})\\
(\ref{eq:ProofPart43}):&\qquad -(\mathop{b}\mathop{b}\mathop{b}) (\mathop{a} \mathop{a}\mathop{a}\mathop{a}^{\displaystyle\overset {u^{-1}}\downarrow}) (\mathop{b}^{\displaystyle\overset {v^{-1}}\downarrow}\mathop{b}\mathop{b})\\
(\ref{eq:ProofPart44}):&\qquad(\mathop{b}\mathop{b} \mathop{b}^{\displaystyle\overset u\downarrow}) (\mathop{a}^{\displaystyle\overset v\downarrow} \mathop{a}\mathop{a}\mathop{a}) (\mathop{b}\mathop{b}\mathop{b})\\
(\ref{eq:ProofPart45}):&\qquad -(\mathop{b}\mathop{b} \mathop{b}) (\mathop{a}^{\displaystyle\overset v\downarrow} \mathop{a}\mathop{a}\mathop{a}) (\mathop{b}^{\displaystyle\overset {u^{-1}}\downarrow}\mathop{b}\mathop{b})\\
(\ref{eq:ProofPart46}):&\qquad -(\mathop{b}\mathop{b} \mathop{b}^{\displaystyle\overset u\downarrow}) (\mathop{a} \mathop{a}\mathop{a}\mathop{a}^{\displaystyle\overset {v^{-1}}\downarrow}) (\mathop{b}\mathop{b}\mathop{b})\\
(\ref{eq:ProofPart47}):&\qquad (\mathop{b}\mathop{b}\mathop{b}) (\mathop{a} \mathop{a}\mathop{a}\mathop{a}^{\displaystyle\overset {v^{-1}}\downarrow}) (\mathop{b}^{\displaystyle\overset {u^{-1}}\downarrow}\mathop{b}\mathop{b})
\end{align*}
\endgroup

So, in total there are 16 sums which contribute to (\ref{eq:AppAntisymmetricity}). Again we will prove that coefficient with any monomial is zero. To do so, construct a graph similar to one described in (\ref{eq:generalquadruples}). In this case we should also take care of relations $uu^{-1}=u^{-1}u=1$ and $vv^{-1}=v^{-1}v=1$ which appears to be quite simple. First we assume that $H_1$ and $H_2$ are already reduced as elements of the cyclic space. So, in particular, they do not contain neighboring $v$ and $v^{-1}$. Next, we should connect quadruples of the form $(*,*,v,v^{-1})$ and $(v^{-1},v,*,*)$ with an edge as well as similar pairs for $(*,*,v^{-1},v)$, $(*,*,u^{-1},u)$, and $(*,*,u,u^{-1})$.

Indeed, say we encounter $a_1a_2\dots a_jb_k\dots b_{k-2}b_{k-1}a_{j+1}\dots a_N$ with $b_{k-1}=v^{-1}$ and $a_{j+1}=v$. Since $b_{k-1}a_{j+1}=v^{-1}v=vv^{-1}=a_{j+1}b_{k-1}$ the same monomial is given by the following combination\newline $a_1a_2\dots a_{j+1}b_{k-1}b_k\dots b_{k-2}a_{j+1}\dots a_N$. Here it is important that this flip is always possible when there is no regular edge, indeed, this would imply that $b_k=a_{j+1}=b_{k-1}^{-1}$ which contradicts the assumption of $H_1$ and $H_2$ are presented in reduced form.

The graph for this case consists of 792 edges, so we do not present it here. However it still satisfies the property that the sum of contributing coefficients over each loop is equal to zero.

Finally, it is worth noticing that one can define a function $f$ on the vertices of the above graph by taking the sum of the coefficients on any path from 0 to the given vertex. (Like in (\ref{eq:ProofPotential})). There is a nice formula for this function
\begin{align*}
f(x_1,x_2,x_3,x_4)=\left\{\begin{array}{ll}
\deg x_2,&x_2=x_4,\\
\deg x_3,&x_3x_4=1,\\
0,&\textrm{otherwise}.
\end{array}
\right.
\end{align*}
Here we assume $\deg v=\deg u^{-1}=+1$ while $\deg u=\deg v^{-1}=-1$.

\bibliographystyle{alpha}
\bibliography{references}

\end{document}